\documentclass[runningheads, a4paper]{llncs}
\usepackage[T1]{fontenc}
\usepackage{graphicx}
\usepackage{times}
\usepackage{silence}
\let\vec\relax
\usepackage{amsmath}
\usepackage{bm,mathrsfs,float,amssymb,subeqnarray,setspace,epstopdf,color}
\usepackage{enumitem}
\usepackage{soul}
\usepackage{wrapfig}
\usepackage{multirow,multicol,makecell}

\newtheorem{theorem*}{Finding}

\newtheorem{assumption}{Assumption}

\newcommand{\F}{\mathbb{F}}
\newcommand{\GL}{\mathrm{GL}}

\newcommand{\C}{\mathcal{C}}

\newcommand{\Hash}{\mathbf{Hash}}
\newcommand{\SamTen}{\mathbf{SamTen}}
\newcommand{\ComSub}{\mathbf{ComSub}}
\usepackage[hypertexnames=false]{hyperref}
\begin{document}
\title{Tensor-Action Ko--Lee Cryptography: A Framework and Structural Cryptanalysis of Commuting-Subgroup Constructions}
\titlerunning{Tensor-Action Ko--Lee Cryptography}

\author{Ziyan Chen\inst{1}, Yuqiao Wang\inst{2}}
\institute{}

\authorrunning{Z. Chen and Y. Wang}
%

\institute{University of Sydney, Australia \and Independent Researcher}

\maketitle

\begin{abstract}
Tensor isomorphism has been studied as an algebraic problem relevant to post-quantum cryptography, while its use in public-key encryption remains open. In this paper, we formulate a Ko--Lee-style framework for public-key encryption from cubic tensor actions and prove its formal correctness. We then show that the framework is generically insecure when the commuting matrix subgroups are given by public finite generating sets. Viewing a cubic tensor as a vector in a $d^3$-dimensional space, a linear decomposition attack recovers the shared tensor from the public transcript in polynomial time without recovering either secret action. We also cryptanalyze three natural commuting-subgroup constructions---field-extension, block-diagonal, and tensor-product constructions---and give toy-scale experiments illustrating their specific structural leakage. Finally, we examine the lower-dimensional leakage caused by scaled-block structure. The contribution is therefore a framework proposal together with its cryptanalysis; it does not provide a secure public-key encryption scheme.
\end{abstract}

\section{Introduction}

Public-key encryption (PKE) is a fundamental primitive in modern cryptography, which enables two parties to communicate securely over an open channel. Classical constructions such as RSA and elliptic-curve cryptography rely on number-theoretic
hardness assumptions \cite{RivestShamirAdleman1978,Miller1986,Koblitz1987}, but Shor's
algorithm shows that factoring and discrete logarithms are vulnerable to quantum attacks
\cite{Shor1994}. This motivates the search for alternative algebraic foundations for
post-quantum public-key cryptography.

One important direction has been non-commutative cryptography. Early proposals include
the Anshel--Anshel--Goldfeld (AAG) scheme and the Ko--Lee cryptosystem, especially over
braid groups \cite{AnshelAnshelGoldfeld1999,KoLeeCheonHanKangPark2000BraidPKE,Dehornoy2000BraidCrypto}. Braid groups were attractive
partly because of their rich combinatorial structure and faithful linear representations
\cite{Bigelow2001,Krammer2002}. Polycyclic groups were later proposed as an alternative
non-commutative platform \cite{EickKahrobaei2004}. However, the subsequent history of
group-based cryptography revealed substantial cryptanalytic obstacles. In braid-based
settings, progress on conjugacy algorithms, including ultra summit set methods, weakened
the original platform assumptions \cite{Gebhardt2006,BirmanGebhardtGonzalezMeneses2008}.
Polynomial-time attacks on the braid Diffie--Hellman conjugacy problem and length-based
attacks against braid-group and polycyclic-group constructions were also developed
\cite{CheonJun2003LDA,MyasnikovUshakov2007LBA,GarberKahrobaeiLam2015}, while
Shpilrain and Ushakov observed that the conjugacy search problem is often neither necessary
nor sufficient for breaking Ko--Lee-type protocols \cite{ShpilrainUshakov2006CSPUnnecessary}.

These developments motivate careful analysis of how visible algebraic structure can undermine
the intended hardness assumption in non-commutative protocols.
In this context, the recent proposal of general linear group actions on tensors is especially
interesting \cite{JiQiaoSongYun2019TensorCrypto}. The natural action  on cubic tensors provides a multilinear setting in which the associated inversion problem,
namely tensor isomorphism, has been proposed as a candidate source of post-quantum hardness
\cite{JiQiaoSongYun2019TensorCrypto}. At the same time, this hardness is highly sensitive to
the structure of the tensor instance: broad classes of tensors with sparse, triangular, or
otherwise rigid support admit efficient orbit recovery or isomorphism testing
\cite{NarayananQiaoTang2024,RanSamardjiska2024}. This makes instance generation a
central issue. To address this, strong-key sampling has recently been proposed as a way to
avoid obviously degenerate tensor instances \cite{NarayananStrongKeys2025}.
Motivated by these works, we ask whether tensor actions can support a public-key encryption
framework. We take the Ko--Lee paradigm as a natural starting point. The essential algebraic
requirement in Ko--Lee-type constructions is the ability to sample secret actions from two
commuting subsets \cite{KoLeeCheonHanKangPark2000BraidPKE}. Translating this idea from group
conjugation to tensor actions leads to a natural tensor-based analogue in which the main
instantiation problem becomes: how can one generate two commuting subgroups of
$\GL_d(\F_p)$ whose induced tensor actions do not immediately leak exploitable
algebraic structure?

In this paper, we propose a Ko--Lee-style framework for tensor-action public-key encryption
and analyze its security. Our main cryptanalytic result shows that, when the commuting
subgroups are described by public finite generating sets, the framework is subject to a
generic linear decomposition attack~\cite{MyasnikovRomankov2014LDA}. The attack works in
the $d^3$-dimensional tensor space and recovers the shared tensor from the public transcript
without recovering the secret matrices. Thus changing only the tensor sampler or the choice
of publicly generated commuting subgroups cannot secure the present linear-action template.

Our generic attack is also closely related to the broader algebraic-span cryptanalysis
framework of Ben-Zvi, Kalka, and Tsaban, which replaces nonlinear group constraints by
computations in the linear spans generated by their matrix representations
\cite{BenZviKalkaTsaban2018}. The attack in Section~\ref{sec:generic-lda} should therefore
be viewed as an application of the classical linear-decomposition/algebraic-span principle
to the present tensor-action Ko--Lee framework.

We additionally study natural candidate constructions
of commuting subgroups in $\GL_d(\F_p)$, including field-extension-based abelian
subgroups, block-diagonal commuting non-abelian subgroups, and tensor-product commuting
subgroups, and show that these obvious choices are insecure. In particular, they expose
diagonal, fixed-subspace, tensor-separable, or lower-dimensional matrix structure that makes
the induced shared-key problem easier than intended. These construction-specific analyses
remain useful because they identify explicit failure mechanisms and support the empirical
evaluation, even though the generic attack already breaks the framework.

The present work should therefore be viewed as a framework proposal together with its
cryptanalysis. It introduces a tensor-action analogue of the Ko--Lee paradigm, proves that the
public finite-generator version is generically insecure, and studies additional failure
mechanisms in several immediate instantiations.

\paragraph{Our contributions.}
The main contributions of the paper are as follows.
\begin{itemize}[leftmargin=*]
    \item We formalize a Ko--Lee-style public-key encryption framework built from the natural action of $\GL_d(\F_p)^3$ on cubic tensors, and prove its basic correctness.
    \item We give a deterministic polynomial-time linear decomposition attack that recovers the shared tensor for any publicly finitely generated pair of commuting action subgroups.
    \item We analyze three natural commuting-subgroup constructions in $\GL_d(\F_p)$, identify their specific structural leakage, and give toy-scale empirical illustrations of the corresponding recovery procedures.
    \item We analyze the lower-dimensional leakage created by scaled-block structure, without presenting this structural reduction as a complete attack in every parameter regime.
\end{itemize}

The contribution is a formal framework and its cryptanalysis, not a secure concrete
instantiation. In particular, the generic attack shows that subgroup and tensor sampling
alone cannot repair the framework while its public linear-action structure is retained.

\paragraph{Roadmap.}
Section~\ref{sec:strong-key-sampling} fixes notation and recalls the heuristic tensor sampler used throughout the paper. Section~\ref{sec:tensor-ko-lee} introduces the Ko--Lee-style tensor-action framework. Section~\ref{sec:commuting-subgroup-instantiation} presents natural candidate constructions for $\ComSub(p,d)$, while Section~\ref{sec:pitfalls-commuting} gives construction-specific analyses. Section~\ref{sec:empirical-illustration} gives toy-scale empirical illustrations of those analyses. Section~\ref{sec:generic-lda} presents the generic linear decomposition attack. Section~\ref{sec:discussion} studies the influence of scaled-block structure. Finally, the appendices collect the proofs and auxiliary details.

\section{Preliminaries}
\label{sec:prelim}

\subsection{Tensors, $\GL$-action, and 3-Tensor Isomorphism}
\paragraph{Cubic tensors.} Fix a prime $p$ and let $\F_p$ be the finite field with $p$ elements. For an integer $d\ge 1$, a \emph{(cubic) 3-tensor} over $\F_p$ is an element:
\[
T \in \F_p^{d}\otimes \F_p^{d}\otimes \F_p^{d} \quad \text{ or }\quad T \in \F_p^{d \times d \times d}
\]
equivalently a three-dimensional array $T=(T_{i,j,k})_{i,j,k\in[d]}$ where $[d]:=\{1,\dots,d\}$. It is often convenient to represent $T$ as a list of slices along the third mode, namely
\[
T=[T_1,\dots,T_d],\qquad T_i\in \F_p^{d\times d}\text{ for }i\in[d].
\]

\paragraph{Mode-wise multiplication.}
For $A,B,C\in \GL_d(\F_p)$, define the natural action \cite{JiQiaoSongYun2019TensorCrypto}:
\[
(A,B,C)\cdot T \in \F_p^{d}\otimes \F_p^{d}\otimes \F_p^{d}
\]
by its entries
\begin{equation}
\label{eq:tensor-action}
\bigl((A,B,C)\cdot T\bigr)_{a,b,c}
:= \sum_{i,j,k\in[d]} A_{a,i}\,B_{b,j}\,C_{c,k}\,T_{i,j,k}.
\end{equation}
This is the standard $\GL_d(\F_p)^3$ action underlying tensor isomorphism.

\paragraph{3-Tensor Isomorphism (3TI).}
Given two tensors $T,T'\in \F_p^{d}\otimes \F_p^{d}\otimes \F_p^{d}$, find $(A,B,C)\in \GL_d(\F_p)^3$ such that $T'=(A,B,C)\cdot T$, assuming such a triple exists.

\subsection{Composition and commutativity of 3-tensor actions}
\label{subsec:tensor-action-composition}

In this subsection we record a useful algebraic property of the
$\GL_d(\F_p)^3$ action on 3-tensors.  Recall that we view
\[
T \in \F_p^{d}\otimes \F_p^{d}\otimes \F_p^{d}
\]
as a list of mode-3 slices
\[
T = [T_1,\dots,T_d],\qquad T_k\in \F_p^{d\times d},
\]
and redefine the action of $(A,B,C)\in \GL_d(\F_p)^3$ by
\begin{equation}
\label{eq:slice-action}
\bigl((A,B,C)\cdot T\bigr)_c
\;=\;
\sum_{k=1}^d C_{c,k}\, A T_k B^{\top},
\qquad c\in[d].
\end{equation}

The following lemma records the basic composition rule.
\begin{lemma}[Composition rule for 3-tensor actions]
\label{lem:tensor-action-composition}
Let $A,B,C,D,E,F\in \GL_d(\F_p)$ and let
$T=[T_1,\dots,T_d]\in \F_p^{d}\otimes \F_p^{d}\otimes \F_p^{d}$. Then, for every $c\in[d]$,
\begin{align*}
&\bigl((A,B,C)\cdot((D,E,F)\cdot T)\bigr)_c\\
= & \bigl((AD,\;BE,\;CF)\cdot T\bigr)_c.
\end{align*}
Equivalently, as operators on tensors,
\[
(A,B,C)\circ(D,E,F) \;=\; (AD,\;BE,\;CF).
\]
\end{lemma}
The componentwise operation follows by expanding Equation~\eqref{eq:slice-action} for the two successive actions and then interchanging the order of summation. As an immediate consequence, we obtain the following corollary.

\begin{corollary}[Commuting action]
\label{cor:commuting-triples}
Let $A,B,C,D,E,F\in \GL_d(\F_p)$. If
\[
AD=DA,\qquad BE=EB,\qquad CF=FC,
\]
then for every tensor $T$,
\[
(A,B,C)\cdot\bigl((D,E,F)\cdot T\bigr)
=
(D,E,F)\cdot\bigl((A,B,C)\cdot T\bigr).
\]
\end{corollary}

\subsection{Heuristic Tensor Sampling}
\label{sec:strong-key-sampling}

The computational hardness of the 3TI problem is highly sensitive to the algebraic structure of the input tensor. For example, tensors with sparse support, block-diagonal structure, or triangular patterns may allow the action of
$\GL_d(\F_p)^3$ to be partially linearized, yielding efficient recovery of the isomorphism
witness \cite{NarayananQiaoTang2024,RanSamardjiska2024}.

Motivated by recent work on strong-key sampling \cite{NarayananStrongKeys2025}, we
consider the following heuristic sampling procedure for cubic tensors
\[
T\in \F_p^{d\times d\times d}.
\]
We emphasize that the rigorous results in \cite{NarayananStrongKeys2025} are developed
for boundary formats, whereas the cubic format considered here is an interior format.
Accordingly, the procedure below should be regarded only as a heuristic way to generate
structured but scrambled cubic tensors, not as a source of security for the encryption
template.

\paragraph{Step 1: Diagonal seed tensor.}
Sample a diagonal tensor
\[
T_0 := \sum_{i=1}^{d} \lambda_i\, e_i \otimes e_i \otimes e_i,
\]
where $\{e_1,\dots,e_d\}$ is the standard basis of $\F_p^d$ and
\[
\lambda_i \xleftarrow{\$} \F_p^\ast
\qquad \text{independently for all } i\in[d].
\]
Here $\xleftarrow{\$}$ denotes uniform sampling from the indicated set.

\paragraph{Step 2: Random invertible scrambling.}
Sample three independent invertible matrices
\[
U,V,W \xleftarrow{\$} \GL_d(\F_p),
\]
and define
\begin{equation}
T := (U,V,W)\cdot T_0.
\end{equation}

\begin{definition}
\label{def:strong-key-sampling}
We write
\[
T \leftarrow \SamTen(p,d)
\]
to denote sampling a tensor $T\in \F_p^{d\times d\times d}$ according to the heuristic procedure above.
\end{definition}

\begin{remark}[Heuristic status in the cubic case]
\label{rem:heuristic-cubic-sampling}
The sampling procedure above is used only heuristically in the cubic $d\times d\times d$ setting. Unlike the boundary-format setting studied in
\cite{NarayananStrongKeys2025}, it does not currently provide a rigorous guarantee of non-degeneracy or cryptographic security in the cubic case. In particular, we do not claim that tensors sampled by Definition~\ref{def:strong-key-sampling} are secure against tensor-action
inversion or related attacks.

Thus, in the present paper, $\SamTen(p,d)$ should be understood only as a convenient
heuristic sampler for discussing the framework. Regardless of the tensor distribution, the
generic attack in Section~\ref{sec:generic-lda} recovers the shared tensor in the present
public linear-action framework. Consequently, improving the tensor sampler alone cannot
make this encryption template secure.
\end{remark}

\section{A Tensor-Action Ko–Lee Framework}
\label{sec:tensor-ko-lee}

\paragraph{Overview and Setup}
To obtain two commuting subgroups for the Ko--Lee-style key agreement, we make
the following assumption.
\begin{assumption}[Public Commuting Subgroups]
\label{ass:comsub}
Given a prime $p$ and a dimension $d$, we assume access to a public procedure
$\ComSub(p,d)$. It outputs two element-wise commuting subgroups
$\mathcal{X},\mathcal{Y}\leq\GL_d(\F_p)$ together with finite generating sets; we write
\[
(\mathcal{X}, \mathcal{Y}) \leftarrow \ComSub(p,d)
\]
\end{assumption}

\paragraph{Scope.}
The construction below is a formal template. Proposition~\ref{prop:correctness} establishes
only that encryption and decryption are algebraically consistent. It does not establish a
security property. Section~\ref{sec:generic-lda} shows that the public finite-generator
version of the template is generically insecure.

In the protocol, Alice and Bob communicate over an insecure channel, where Alice is the receiver and Bob is the sender. Let $p \in \mathbb{N}$ be prime, $\F_p$ be the base field, and $d$ denote the tensor dimension. To sample the base tensor used in the protocol, we recall the heuristic tensor sampler from Definition~\ref{def:strong-key-sampling}:
\[
T \leftarrow \SamTen(p,d)
\]
In the protocol below, Alice samples her secret action from $\mathcal X$ and
Bob samples the ephemeral action from $\mathcal Y$. In the classical Ko--Lee setting in Appendix~\ref{sec:ko-lee-template}, the intended hard problem is based on conjugation. In our framework, the analogous role is played by tensor-action inversion. As emphasized earlier, however, the difficulty of inverting the action depends strongly on the tensor family and on the subgroup structure used in the instantiation \cite{JiQiaoSongYun2019TensorCrypto}.

\subsubsection*{Key generation}
Alice, the receiver, samples:
\begin{enumerate}
    \item $T \leftarrow \SamTen(p,d)$ as in Section~\ref{sec:strong-key-sampling}.
    \item Independently sample $A,B,C \leftarrow \mathcal{X}$.
\end{enumerate}
Output the secret key and public key:
\[
sk := (A,B,C),
\qquad
pk := \bigl(T,\ (A,B,C)\cdot T\bigr).
\]

\subsubsection*{Encryption}
To encrypt a message $m\in\{0,1\}^\ell$ for Alice under $pk=(T,(A,B,C)\cdot T)$, Bob does the following:
\begin{enumerate}
    \item Sample $D,E,F \leftarrow \mathcal{Y}$.
    \item Compute the shared tensor
    \[
    S := (D,E,F)\cdot\bigl((A,B,C)\cdot T\bigr).
    \]
    and hash it:
    \[
    h := \Hash(S)\in\{0,1\}^\ell.
    \]
    \item Send the ciphertext to Alice:
    \[
    \C := (U,c) := \bigl((D,E,F)\cdot T,\ m\oplus h\bigr).
    \]
\end{enumerate}

\subsubsection*{Decryption}
Given $\C=(U,c)$ and $sk=(A,B,C)$, Alice computes the shared tensor
\[
S' := (A,B,C)\cdot U,
\]
and hashes it to recover the raw message $m'$ as
\[
h' := \Hash(S'),
\qquad
m' := c\oplus h',
\]
where $\oplus$ denotes bitwise exclusive OR.

\subsubsection*{Correctness}
We now prove the correctness of our Ko--Lee tensor PKE scheme.
\begin{proposition}[Correctness of Ko--Lee tensor PKE]
\label{prop:correctness}
For all messages $m\in\{0,1\}^\ell$, decryption outputs $m'=m$.
\end{proposition}

\begin{proof}
By construction, elements of $\mathcal{X}$ commute with elements of $\mathcal{Y}$. Thus,
for $A,B,C\in\mathcal{X}$ and $D,E,F\in\mathcal{Y}$ we have
\[
AD=DA,\qquad BE=EB,\qquad CF=FC.
\]
Therefore Corollary~\ref{cor:commuting-triples} yields
\begin{align*}
    S'&=(A,B,C)\cdot\bigl((D,E,F)\cdot T\bigr)\\
    & =(D,E,F)\cdot\bigl((A,B,C)\cdot T\bigr)=S.
\end{align*}

Thus $h'=\Hash(S')=\Hash(S)=h$, and
\[
m' = c\oplus h' = (m\oplus h)\oplus h = m.
\]
\end{proof}

\begin{remark}[Limitation of the linear-action template]
\label{rem:linear-action-limitation}
The correctness proof uses only commutativity of the two action families. The same
commutativity, together with their public linear representations, enables the generic linear
decomposition attack in Section~\ref{sec:generic-lda}. Therefore no choice of
$\SamTen(p,d)$ and publicly finitely generated $\ComSub(p,d)$ can provide security within
the present template. A potentially secure variant would have to change at least one of the
features used by the attack, rather than only replacing the concrete samplers; designing such
a different framework is outside the scope of this paper.
\end{remark}

\section{Instantiations of Commuting Subgroups in $\GL_d(\F_p)$}
\label{sec:commuting-subgroup-instantiation}

The tensor PKE framework in Section~\ref{sec:tensor-ko-lee} requires an instantiation of
the commuting-subgroup generation procedure $\mathsf{ComSub}(p,d)$. In particular, the
formal correctness argument needs two commuting subgroups in $\GL_d(\F_p)$. In this
section, we record several natural constructions and their public descriptions. None of them
provides security: Section~\ref{sec:pitfalls-commuting} gives construction-specific structural analyses,
and Section~\ref{sec:generic-lda} later gives an attack on the public finite-generator
framework independently of the subgroup choice.

\subsection{Construction via field extensions (Singer-type embeddings)}
\label{subsec:commuting-field-extension}

A standard and conceptually clean way to generate commuting subgroups of
$\GL_d(\F_p)$ is to embed the multiplicative group of a degree-$d$
extension field into $\GL_d(\F_p)$ via $\mathbb{F}_p$-linear multiplication
maps. This construction is widely known as a \emph{Singer cycle} (or, more generally, a
\emph{field-multiplication subgroup}). In this way one obtains an abelian subgroup,
say $\mathcal{H}$, and two commuting subgroups $\mathcal{X}$ and $\mathcal{Y}$ may then
be chosen as subgroups of $\mathcal{H}$. We now describe how such an abelian subgroup
arises inside $\GL_d(\F_p)$ from field multiplication.

\paragraph{Algebra embedding.}
Let $\mathbb{K}=\mathbb{F}_{p^d}$ and fix an $\mathbb{F}_p$-basis
$\mathcal{B}=(\omega_1,\dots,\omega_d)$ of $\mathbb{K}$. For each $\alpha\in\mathbb{K}$,
consider the $\mathbb{F}_p$-linear endomorphism
\[
m_\alpha:\mathbb{K}\to\mathbb{K},\qquad x\mapsto \alpha x.
\]
Let $\Phi(\alpha)\in \mathrm{Mat}_d(\mathbb{F}_p)$ be the matrix of $m_\alpha$ in the basis
$\mathcal{B}$, that is, the unique matrix satisfying
\[
[\alpha x]_{\mathcal{B}}=\Phi(\alpha)\,[x]_{\mathcal{B}}\qquad\forall x\in\mathbb{K},
\]
where $[\cdot]_{\mathcal{B}}\in \mathbb{F}_p^d$ denotes the coordinate vector in the basis
$\mathcal{B}$. The following lemma states that this construction indeed yields an
abelian subgroup of $\GL_d(\F_p)$.

\begin{lemma}[Field embedding into matrices]
\label{lem:field-embed}
The map $\Phi:\mathbb{K}\to \mathrm{Mat}_d(\mathbb{F}_p)$ is an injective
$\mathbb{F}_p$-algebra homomorphism:
\[
\Phi(\alpha+\beta)=\Phi(\alpha)+\Phi(\beta),\qquad
\Phi(\alpha\beta)=\Phi(\alpha)\Phi(\beta),\qquad
\Phi(1)=I_d.
\]
Moreover, $\alpha\neq 0$ if and only if $\Phi(\alpha)\in \GL_d(\F_p)$, and
\[
\Phi(\mathbb{K}^\times)\le \GL_d(\F_p)
\quad\text{is abelian.}
\]
\end{lemma}

The commutativity is immediate from field multiplication: $m_\alpha m_\beta=m_{\alpha\beta}=m_{\beta\alpha}=m_\beta m_\alpha$.

After Lemma~\ref{lem:field-embed}, we define the abelian (indeed cyclic) subgroup
\[
\mathcal{H}:=\Phi(\mathbb{K}^\times)\le \GL_d(\F_p),
\qquad |\mathcal{H}|=p^d-1.
\]
Then any two subgroups $X,Y\le \mathcal{H}$ commute elementwise. Concretely, one may set
\[
X=\langle \Phi(\alpha)\rangle,\qquad Y=\langle \Phi(\beta)\rangle,
\]
for randomly sampled $\alpha,\beta\in\mathbb{K}^\times$ (or, more generally, choose
$X,Y$ as random subgroups of $\mathcal{H}$). Since $\mathcal{H}$ is abelian, for all $x\in X$ and
$y\in Y$ we have
\[
xy=yx.
\]
This yields an immediate and efficient sampler for commuting private keys: sample
$\alpha$ uniformly in $\mathbb{K}^\times$, compute $\Phi(\alpha)\in \GL_d(\F_p)$,
and then, if desired, restrict to a prescribed subgroup by exponentiation.

\paragraph{Concrete instantiation via an irreducible polynomial.}
A convenient explicit realization avoids constructing $\mathbb{K}$ abstractly. Choose an
irreducible polynomial $f(t)\in\mathbb{F}_p[t]$ of degree $d$, and let
$C_f\in \mathrm{Mat}_d(\mathbb{F}_p)$ be its companion matrix. Then
\[
\mathbb{F}_p[C_f]\cong \mathbb{F}_p[t]/(f(t))\cong \mathbb{F}_{p^d},
\]
and the unit group $\mathbb{F}_p[C_f]^\times$ is a field-multiplication subgroup of
$\GL_d(\F_p)$. In particular,
\[
\mathcal{H}=\langle C_f\rangle \le \GL_d(\F_p)
\quad\text{has order }p^d-1\text{ when }C_f\text{ is a Singer generator.}
\]

\subsection{A block-diagonal construction of commuting non-abelian subgroups}
\label{subsec:block-diagonal-commuting}

Another natural way to construct two commuting subgroups of
$\GL_d(\F_p)$ is to split the ambient space into two complementary
coordinate subspaces of equal dimension and let each subgroup act nontrivially on only
one half. Let
\[
d=2m.
\]
We define
\[
X:=\left\{
\begin{pmatrix}
\bar X & 0\\
0 & I_m
\end{pmatrix}
:\bar X\in \GL_m(\F_p)
\right\},
\qquad
Y:=\left\{
\begin{pmatrix}
I_m & 0\\
0 & \bar Y
\end{pmatrix}
:\bar Y\in \GL_m(\F_p)
\right\}.
\]
Thus, elements of $X$ act nontrivially only on the first $m$ coordinates, while elements
of $Y$ act nontrivially only on the last $m$ coordinates. In particular, if
\[
A=\begin{pmatrix}\bar A&0\\0&I_m\end{pmatrix},\quad
B=\begin{pmatrix}\bar B&0\\0&I_m\end{pmatrix},\quad
C=\begin{pmatrix}\bar C&0\\0&I_m\end{pmatrix}\in X,
\]
and
\[
D=\begin{pmatrix}I_m&0\\0&\bar D\end{pmatrix},\quad
E=\begin{pmatrix}I_m&0\\0&\bar E\end{pmatrix},\quad
F=\begin{pmatrix}I_m&0\\0&\bar F\end{pmatrix}\in Y,
\]
with
\[
\bar A,\bar B,\bar C,\bar D,\bar E,\bar F\in \GL_m(\F_p),
\]
then every element of $X$ commutes with every element of $Y$. We also claim that this
construction yields two commuting subgroups, while each subgroup may itself be
non-commutative.

\begin{lemma}
\label{lem:block-diagonal-commuting}
The sets $X$ and $Y$ defined above are subgroups of $\GL_{2m}(\F_p)$. Moreover:
\begin{enumerate}
    \item for every $x\in X$ and $y\in Y$, one has $xy=yx$;
    \item if $m\ge 2$, then $X$ is non-abelian and $Y$ is non-abelian.
\end{enumerate}
\end{lemma}

The cross-commutativity follows by block multiplication, since
\[
\mathrm{diag}(\bar X,I_m)\mathrm{diag}(I_m,\bar Y)
=\mathrm{diag}(\bar X,\bar Y)
=\mathrm{diag}(I_m,\bar Y)\mathrm{diag}(\bar X,I_m).
\]

This construction is simple and explicit: it guarantees elementwise commutativity between
the two subgroups while preserving non-commutativity inside each subgroup. On the other
hand, it also introduces a very visible direct-sum decomposition of the ambient space
\[
\mathbb{F}_p^{2m}=\mathbb{F}_p^m\oplus \mathbb{F}_p^m,
\]
which may leak structural information and therefore should be treated with caution in
cryptographic applications.

\subsection{Construction via tensor-product decomposition}
\label{subsec:tensor-product-commuting}

When the ambient dimension is a square, say
\[
d=n^2,
\]
another natural construction comes from a tensor-product decomposition of the ambient
space. Let
\[
V=\mathbb{F}_p^n,
\qquad
W=\mathbb{F}_p^n,
\qquad
U:=V\otimes W\cong \mathbb{F}_p^{n^2}.
\]
Inside $\GL(U)\cong \GL_{n^2}(\F_p)$, define
\[
X:=\{\,M\otimes I_n : M\in \GL_n(\F_p)\,\},
\qquad
Y:=\{\,I_n\otimes N : N\in \GL_n(\F_p)\,\}.
\]
By the mixed-product rule for Kronecker products, for all
$M,N\in\GL_n(\F_p)$ one has
\[
(M\otimes I_n)(I_n\otimes N)
=
M\otimes N
=
(I_n\otimes N)(M\otimes I_n).
\]
Thus every element of $X$ commutes with every element of $Y$. The following lemma records
the subgroup properties.

\begin{lemma}
\label{lem:tensor-product-commuting}
The sets $X$ and $Y$ defined above are subgroups of $\GL_{n^2}(\F_p)$. Moreover:
\begin{enumerate}
    \item for every $x\in X$ and $y\in Y$, one has $xy=yx$;
    \item if $n\ge 2$, then $X$ is non-abelian and $Y$ is non-abelian.
\end{enumerate}
\end{lemma}

The cross-commutativity follows from the mixed-product rule: $(M\otimes I_n)(I_n\otimes N)=M\otimes N=(I_n\otimes N)(M\otimes I_n)$.

This construction is attractive because it preserves substantial non-commutative freedom
inside each subgroup while making cross-commutativity automatic. However, it also exposes
the public tensor-product decomposition
\[
\mathbb{F}_p^{n^2}\cong \mathbb{F}_p^n\otimes \mathbb{F}_p^n,
\]
and, as we explain in Subsection~\ref{subsec:tensor-product-separability}, this separability makes
the resulting shared-key computation easy.
\section{Construction Pitfalls of Commuting Subgroups in $\GL_d(\F_p)$}
\label{sec:pitfalls-commuting}

In this section, we analyze the constructions introduced in
Section~\ref{sec:commuting-subgroup-instantiation} and explain the additional structural
weaknesses that they expose. The generic linear decomposition attack in
Section~\ref{sec:generic-lda} already breaks the public finite-generator framework
independently of these choices. The analyses below are nevertheless informative because
they identify more explicit construction-specific leakage and motivate the recovery
procedures evaluated on toy instances in Section~\ref{sec:empirical-illustration}.

From the attacker's perspective, the public transcript of the tensor Ko--Lee scheme consists
of
\[
T,
\qquad
T_A:=(A,B,C)\cdot T,
\qquad
T_B:=(D,E,F)\cdot T,
\]
where $(A,B,C)$ is Alice's secret action and $(D,E,F)$ is Bob's ephemeral action. The
attacker does not need to recover all six secret matrices. To break the session, it is
enough to compute the shared tensor
\[
S=(D,E,F)\cdot T_A=(A,B,C)\cdot T_B=(AD,BE,CF)\cdot T,
\]
or enough to recover either secret action and then apply it to the other public tensor.
Thus the concrete cryptanalytic task is a shared-key recovery problem, which may be reduced
either to tensor-action inversion for $T_A$ or $T_B$, or to a direct computation of the
combined action on $T$.

The analyses below exploit visible algebraic structure in the candidate commuting
subgroups. Depending on the construction and the stated hypotheses, this structure may
enable action recovery, reduce the shared-key task to an auxiliary problem, or permit direct
recovery of the shared tensor.



Table~\ref{tab:pitfall-taxonomy} summarizes the three constructions from
Section~\ref{sec:commuting-subgroup-instantiation}, the visible structure each one exposes, the corresponding attack mechanism, and the resulting design lesson.

\begin{table}[t]
\centering
\small
\caption{Taxonomy of the failure modes of the commuting-subgroup constructions introduced in Section~\ref{sec:commuting-subgroup-instantiation}.}
\label{tab:pitfall-taxonomy}
\renewcommand{\arraystretch}{1.18}
\begin{tabular}{@{}>{\raggedright\arraybackslash}p{0.24\textwidth}>{\raggedright\arraybackslash}p{0.19\textwidth}>{\raggedright\arraybackslash}p{0.24\textwidth}>{\raggedright\arraybackslash}p{0.21\textwidth}@{}}
\hline
Construction & Visible structure & Attack mechanism & Lesson \\
\hline
Field-extension construction (Subsection~\ref{subsec:commuting-field-extension})
& simultaneous diagonalization of extension field
& conditional coordinate-wise action recovery under sufficient tensor support
& avoid efficiently diagonalizable commuting subgroups \\
Block-diagonal construction (Subsection~\ref{subsec:block-diagonal-commuting})
& large fixed or invariant subspace
& reduction to a lower-dimensional mixed-product problem
& avoid large public invariant subspaces \\
Tensor-product construction (Subsection~\ref{subsec:tensor-product-commuting})
& public coarse/fine tensor-factor decomposition
& flatten the action into an explicit left--right matrix sandwich and compute the shared tensor directly
& avoid public tensor-factor decompositions and separable actions \\
\hline
\end{tabular}
\renewcommand{\arraystretch}{1}
\end{table}

\subsection{Conditional Recovery for Simultaneously Diagonalizable Subgroups}
\label{subsec:pitfall-simul-conj}

One natural way to choose two commuting subgroups $\mathcal{X}$ and $\mathcal{Y}$ is
to take them equal to the same abelian subgroup, say $\mathcal{H}$. A known simultaneous
diagonalization exposes coordinate-wise multiplicative structure; when the transformed
tensor has sufficient support, tensor-action inversion reduces to a low-degree
reconstruction problem. We first record the formal definition.

\begin{definition}[Simultaneous diagonalizability]
Let $\mathcal{H}\le \GL_d(\F_p)$. We say that subgroup $\mathcal{H}$ is
\emph{simultaneously diagonalizable over $\mathbb{F}_{p^m}$} if there exist an integer
$m\ge 1$ and a matrix $P\in \GL_d(\mathbb{F}_{p^m})$ such that for every
$h\in \mathcal{H}$,
\[
P^{-1} h P \ \text{is diagonal in}\ \mathrm{Mat}_d(\mathbb{F}_{p^m}).
\]
\end{definition}

We allow the conjugating matrix $P$ to lie in $\GL_d(\mathbb{F}_{p^m})$ (for some
$m\ge 1$) rather than restricting to $\GL_d(\F_p)$ because many natural
matrix families over $\mathbb{F}_p$ admit a diagonal form only after passing to an
extension field. Concretely, a matrix in $\GL_d(\F_p)$ may have an
irreducible characteristic polynomial over $\mathbb{F}_p$, and hence no eigenvalues or
eigenbasis over the base field. After passing to a finite extension $\mathbb{F}_{p^m}$
where the polynomial splits, diagonalization may become possible. This extension-of-scalars
viewpoint is standard in representation theory and is essential for capturing common
cryptographic constructions such as field-multiplication subgroups (Singer cycles), which
are defined over $\mathbb{F}_p$ but diagonalize over $\mathbb{F}_{p^d}$.

We now record the basic structural consequence that such a subgroup is abelian, and then
state a conditional action-recovery result. The proofs are provided in
Appendix~\ref{proof:simul-commute} and Appendix~\ref{proof:diagonalizable-easy}.

\begin{lemma}[Simultaneous implies commutativity]
\label{lem:simul-commute}
Let $\mathcal{H}\subseteq \GL_d(\F_p)$ be a set of matrices. Assume there
exist a field extension $\mathbb{F}_{p^m}$ and a matrix $P\in \GL_d(\mathbb{F}_{p^m})$
such that for every $h\in \mathcal{H}$, the conjugate $P^{-1}hP$ lies in a commuting subset
$\mathcal{K}\subseteq \GL_d(\mathbb{F}_{p^m})$. Then $\mathcal{H}$ is commuting; namely, for all
$h_1,h_2\in \mathcal{H}$ we have $h_1h_2=h_2h_1$.
\end{lemma}

\begin{proposition}[Diagonalizable $\mathcal{H}$ yields efficient action recovery under sufficient support]
\label{prop:diagonalizable-easy}
Assume that $\mathcal{H}$ is simultaneously diagonalizable over
$\mathbb{F}_{p^m}$ with a known conjugator
$P\in\GL_d(\mathbb{F}_{p^m})$. Let
\[
\widetilde T := (P^{-1},P^{-1},P^{-1})\cdot T.
\]
Suppose that $\widetilde T$ has sufficient support for coordinate-wise recovery; for
example, there exists a pivot $(i_0,j_0,k_0)$ such that
\[
\widetilde T_{i_0j_0k_0}\neq 0,
\]
and, for every $i,j,k$, the entries
\[
\widetilde T_{ij_0k_0},\qquad
\widetilde T_{i_0jk_0},\qquad
\widetilde T_{i_0j_0k}
\]
needed in the ratio reconstruction are nonzero. Then an action
$(A,B,C)\in\mathcal{H}^3$ can be recovered from $T$ and $(A,B,C)\cdot T$ in
polynomial time, up to the natural scalar gauge freedom.
\end{proposition}

\begin{corollary}[Conditional action recovery for the field-extension construction]
\label{cor:field-extension-easy}
Let $\mathcal{H}=\Phi(\mathbb{K}^{\times})\le \GL_d(\F_p)$ be the
field-multiplication subgroup from Subsection~\ref{subsec:commuting-field-extension},
and let $P$ be a known simultaneous diagonalizer over a splitting field. If
$\widetilde T=(P^{-1},P^{-1},P^{-1})\cdot T$ satisfies the support condition of
Proposition~\ref{prop:diagonalizable-easy}, then any action
$(A,B,C)\in\mathcal{H}^3$ can be recovered from $T$ and $(A,B,C)\cdot T$ in
polynomial time, up to the natural scalar gauge freedom.
\end{corollary}

This follows because, in the irreducible-polynomial realization $\mathcal{H}=\mathbb{F}_p[C_f]^\times$, the companion matrix $C_f$ becomes diagonalizable over $\mathbb{F}_{p^d}$ once $f$ is split there. Every element of $\mathcal{H}$ is a polynomial in $C_f$ and is therefore diagonalized by the same conjugator. Moreover, that conjugator can be computed efficiently from the roots of $f$ and the associated eigenbasis; see, for example, \cite[Chapter~3]{LidlNiederreiter1997}.

If one chooses $\mathcal{X}=\mathcal{Y}=\mathcal{H}$, an efficiently computable
simultaneous diagonalization exposes coordinate-wise multiplicative structure. Under the
support condition in Proposition~\ref{prop:diagonalizable-easy}, this structure gives
polynomial-time action recovery. We do not claim that the tensor sampler always produces
the required support after the basis change. Independently of this conditional result, the
public finite-generator field-extension construction is broken by the generic linear
decomposition attack in Section~\ref{sec:generic-lda}.

\subsection{Reduction of fixed subspace}

The discussion above concerns cases where tensor inversion itself may become easy. However, solving the full tensor-action inversion problem is not necessary for breaking
the protocol: it suffices to recover the shared tensor used to derive the session key.

We now define an auxiliary problem and explain how the shared key computation for the commuting construction in Subsection~\ref{subsec:block-diagonal-commuting} reduces to it.

\begin{definition}[Auxiliary mixed-product problem.]
\label{def:aux-mixed-product-problem}
Given matrices
\[
S_1,\dots,S_m\in \mathbb{F}_p^{m\times m},
\]
together with
\[
U_i=\sum_{k=1}^m \bar C_{ik}\,\bar A S_k,\qquad
V_i=S_i\bar E^\top
\qquad (1\le i\le m),
\]
we need to compute
\[
W_i=\sum_{k=1}^m \bar C_{ik}\,\bar A S_k \bar E^\top
\qquad (1\le i\le m).
\]
\end{definition}

\begin{proposition}
\label{prop:block-reduction}
Solving the Ko--Lee tensor scheme in Section~\ref{sec:tensor-ko-lee} under the construction in
Subsection~\ref{subsec:block-diagonal-commuting} can be reduced to solving the auxiliary mixed
product problem in Definition~\ref{def:aux-mixed-product-problem}.
\end{proposition}

Proposition~\ref{prop:block-reduction} shows that the block-diagonal construction exposes
an additional lower-dimensional mixed-product structure. Thus the tensor shared-key
problem is reduced to a more structured matrix problem. We do not claim that the
auxiliary mixed-product problem is polynomial-time solvable in complete generality.
Nevertheless, the concrete construction is already insecure under the generic linear
decomposition attack of Section~\ref{sec:generic-lda}, while
Proposition~\ref{prop:block-reduction} identifies an additional construction-specific
source of algebraic leakage. We return in Section~\ref{subsec:influence_scaled-block} to
the more nuanced question of what happens when the relevant fixed or scaled subspace is
small.

\subsection{Tensor-product separability}
\label{subsec:tensor-product-separability}

The previous discussion shows that public direct-sum structure can be dangerous. A different
but equally problematic source of leakage is a public tensor-product decomposition. Assume
that
\[
d=n^2,
\]
and consider the construction from Subsection~\ref{subsec:tensor-product-commuting}. Thus Alice
samples
\[
A=M_1\otimes I_n,
\qquad
B=M_2\otimes I_n,
\qquad
C=M_3\otimes I_n
\]
from $X$, while Bob samples
\[
D=I_n\otimes N_1,
\qquad
E=I_n\otimes N_2,
\qquad
F=I_n\otimes N_3
\]
from $Y$, where $M_1,M_2,M_3,N_1,N_2,N_3\in \GL_n(\F_p)$.

\begin{proposition}[Coarse--fine flattening attack]
\label{prop:coarse-fine-flattening-attack}
Consider the tensor-action Ko--Lee framework in dimension $d=n^2$ under the tensor-product
construction of Subsection~\ref{subsec:tensor-product-commuting}. Given the public tensors
\[
T,
\qquad
T_A:=(A,B,C)\cdot T,
\qquad
T_B:=(D,E,F)\cdot T,
\]
one can compute the shared tensor
\[
T_{AB}:=(A,B,C)\cdot T_B
=
(D,E,F)\cdot T_A
\]
in polynomial time. Hence this tensor-product construction is not suitable as a secure
instantiation of the tensor-action Ko--Lee framework.
\end{proposition}

The proof is given in Appendix~\ref{proof:coarse-fine-flattening-attack}. The key point is
that the public decomposition $\mathbb{F}_p^{n^2}\cong \mathbb{F}_p^n\otimes \mathbb{F}_p^n$
induces a coarse--fine flattening under which Alice's action becomes left multiplication by
$M_1\otimes M_2\otimes M_3$, while Bob's action becomes right multiplication by
$(N_1\otimes N_2\otimes N_3)^\top$. The shared tensor is therefore reduced to a matrix
sandwich that can be computed directly from the public transcript.

\begin{remark}[Interpretation]
\label{rem:tensor-product-separability-interpretation}
The vulnerability here is not the presence of a fixed subspace, but the publicly visible
separation between coarse and fine indices. Once that decomposition is known, the tensor
action collapses to an explicit left--right matrix action after flattening. In this sense,
tensor-product separability exposes an explicit linear-algebraic recovery procedure and is
therefore unsuitable for the proposed framework.
\end{remark}
\section{Empirical Illustration of Construction-Specific Recovery}
\label{sec:empirical-illustration}

We empirically illustrate the construction-specific recovery procedures motivated by
Section~\ref{sec:pitfalls-commuting}
on the three commuting-subgroup constructions from
Section~\ref{sec:commuting-subgroup-instantiation}. In each experiment, the public tensor
$T\leftarrow \SamTen(p,d)$ is sampled using the heuristic sampler from
Section~\ref{sec:strong-key-sampling}; the two commuting subgroups are instantiated as in
the corresponding field-extension, block-diagonal, or tensor-product construction; and
secret elements are sampled from these subgroups to form the public tensors
$T_A=(A,B,C)\cdot T$ and $T_B=(D,E,F)\cdot T$ together with the shared tensor
$S=(D,E,F)\cdot T_A=(A,B,C)\cdot T_B$. The recovery routines are then given only the
public transcript $(T,T_A,T_B)$ and use the structural information exposed by each
construction to compute the shared tensor. Thus these experiments should be read only as
toy-scale demonstrations of the attacks, not as a proof of asymptotic performance or as a
security claim about subgroup choices. In particular, the block-diagonal experiment tests
instances for which the implementation can solve the auxiliary mixed-product problem; it
does not establish a polynomial-time solution to that problem in complete generality. For
the dimension-sweep plots below, we use the
visualization driver with fixed prime $p=101$ and random seed $2026$, and report the
median recovery time over $10$ independent trials for each parameter value. The swept
parameters are $d\in\{3,\dots,12\}$ for the field-extension construction,
$m\in\{2,\dots,11\}$ for the block-diagonal construction with ambient dimension
$d=2m$, and $n\in\{2,\dots,7\}$ for the tensor-product construction with ambient
dimension $d=n^2$; the block-diagonal sampler allows up to $100$ resampling attempts to
obtain the required invertible blocks in the initial tensor.

\begin{figure}[H]
\centering
\begin{minipage}{0.32\textwidth}
\centering
\includegraphics[width=\linewidth]{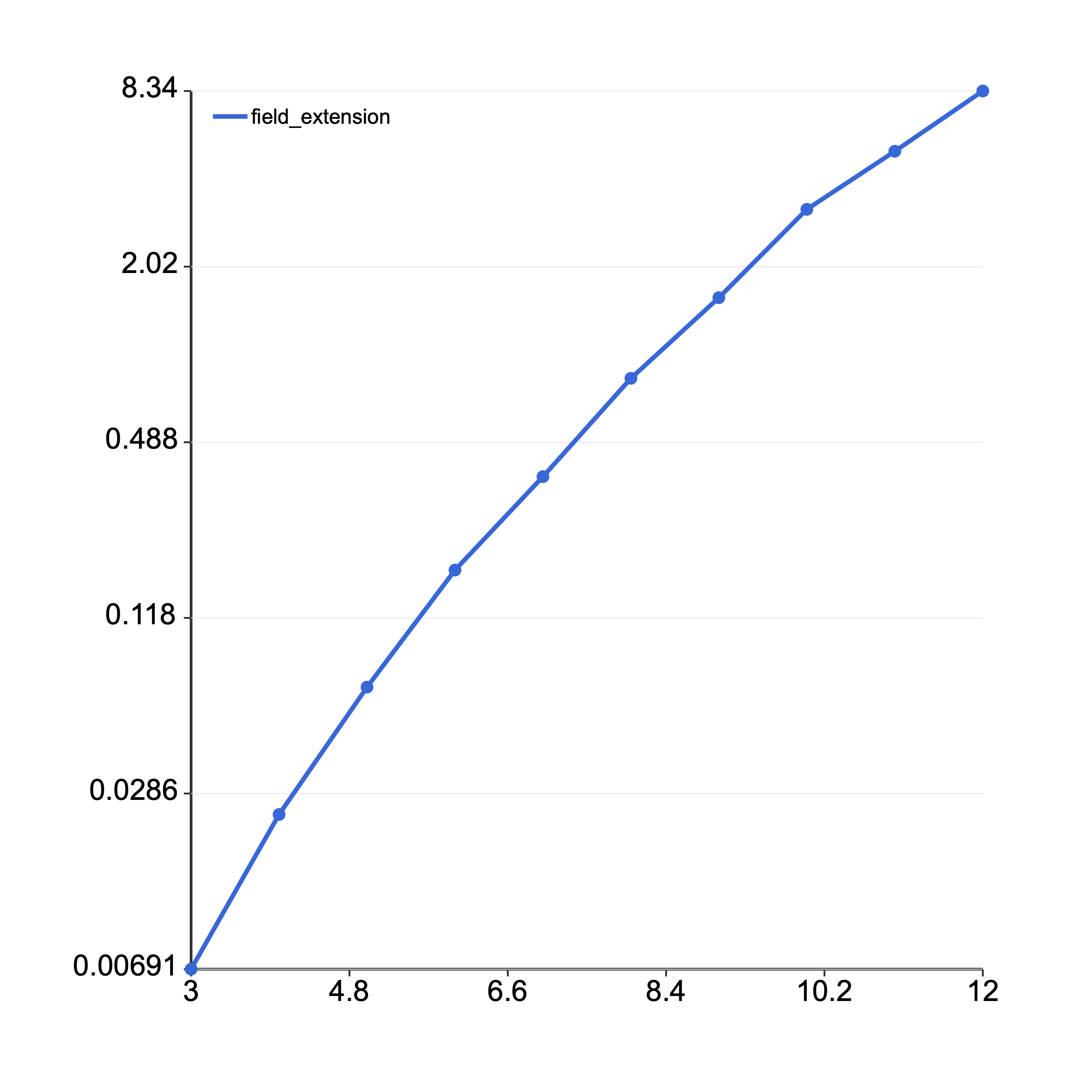}
{\scriptsize (a) Field extension, parameter $d$.}
\end{minipage}\hfill
\begin{minipage}{0.32\textwidth}
\centering
\includegraphics[width=\linewidth]{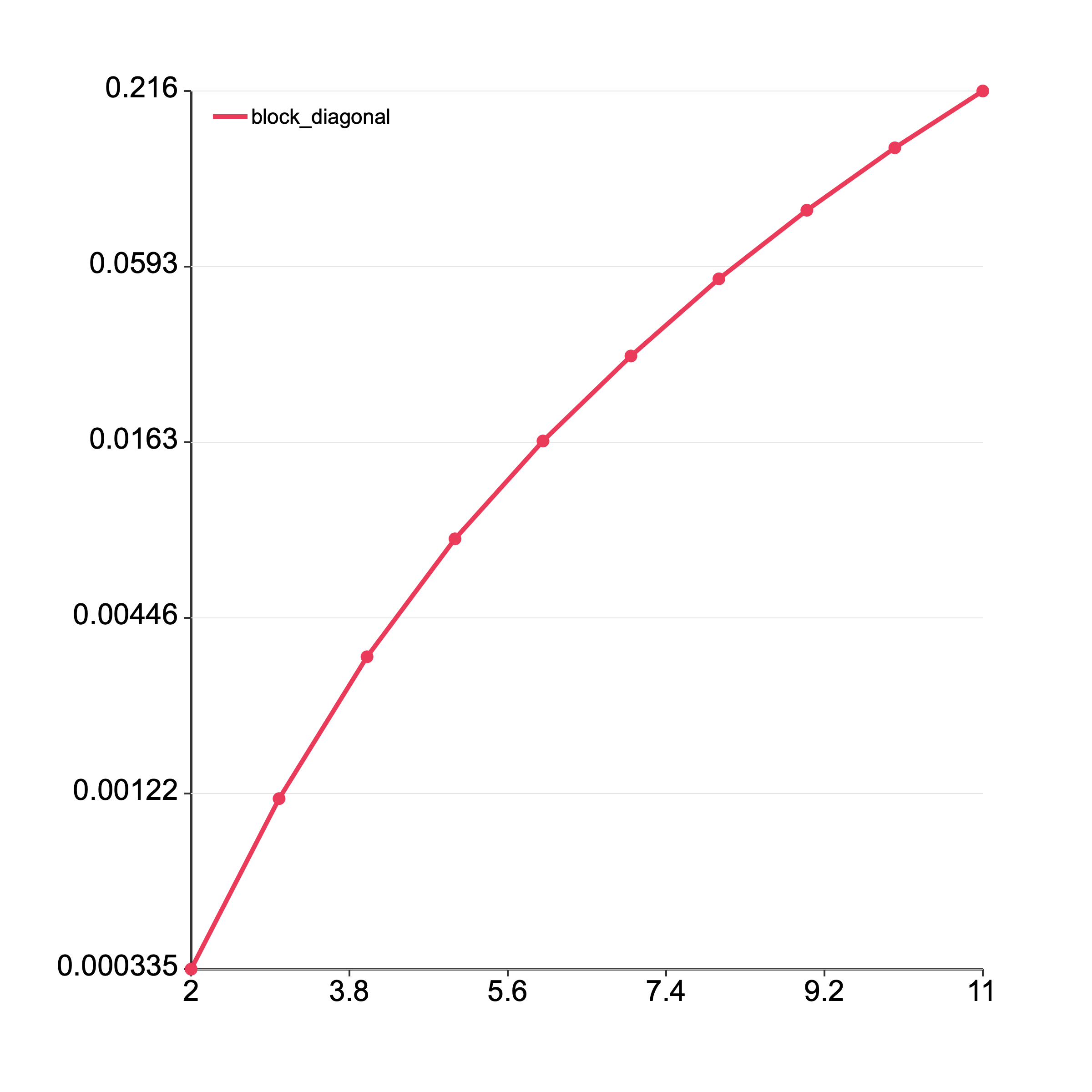}
{\scriptsize (b) Block diagonal, parameter $m$.}
\end{minipage}\hfill
\begin{minipage}{0.32\textwidth}
\centering
\includegraphics[width=\linewidth]{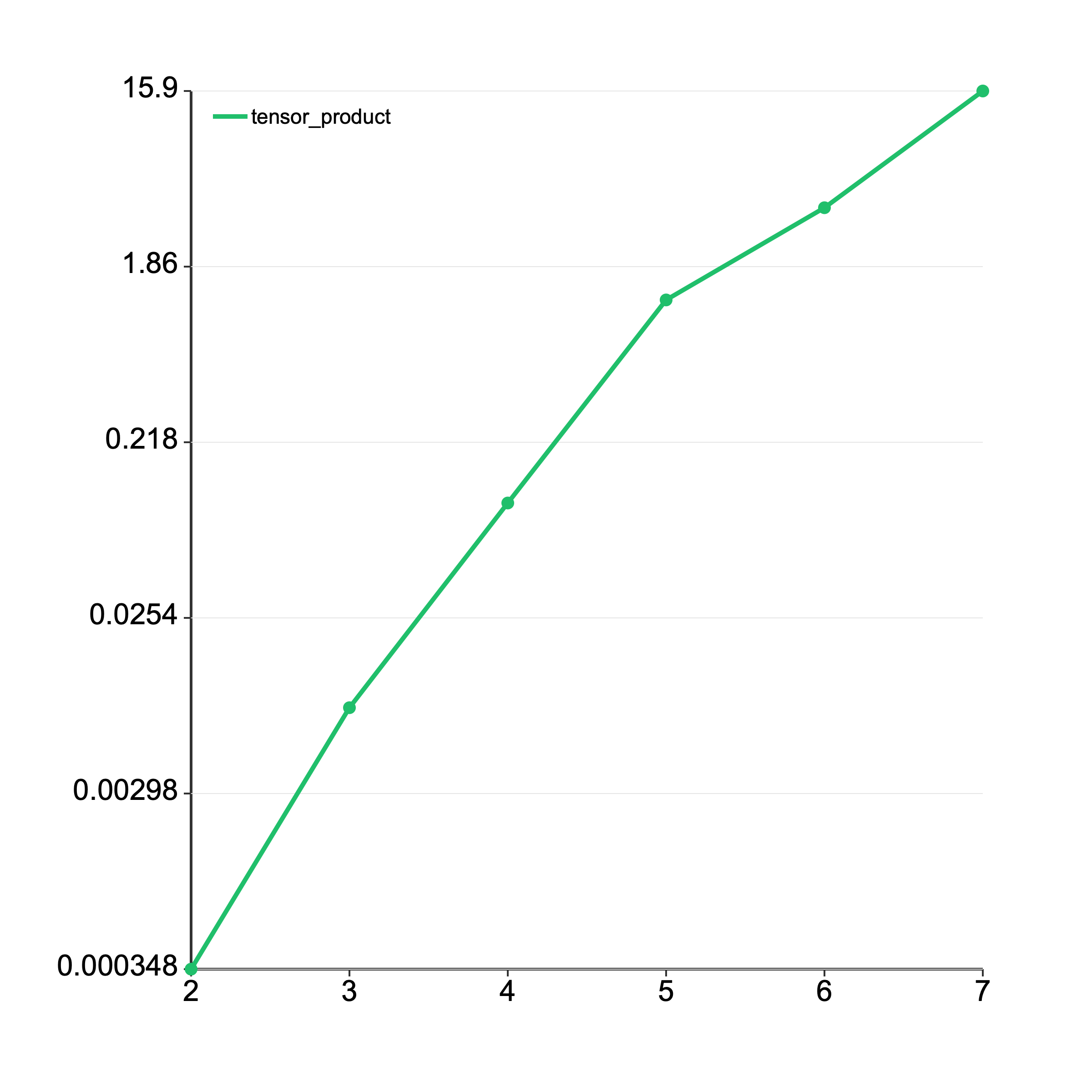}
{\scriptsize (c) Tensor product, parameter $n$.}
\end{minipage}
\caption{Median running time, on a logarithmic scale, for recovering the shared tensor as
the relevant dimension parameter increases. For the block-diagonal construction the matrix
dimension is $d=2m$, while for the tensor-product construction it is $d=n^2$.}
\label{fig:empirical-runtime-dimension}
\end{figure}

Figure~\ref{fig:empirical-runtime-dimension} is consistent with the qualitative behavior
predicted by the structural analysis: as the relevant dimension parameter grows, the
running time of the implemented recovery procedure also increases. This is expected because
the attacks still perform concrete linear-algebraic operations over larger ambient spaces,
even though the subgroup structure makes recovery feasible at these toy dimensions. In the
plotted runs, the recovered tensor agreed with the true shared tensor in every trial, giving
a $100\%$ shared-tensor recovery success rate. We emphasize that this observation is an
empirical sanity check for the specific implementations and parameters considered here; it
does not by itself establish a general complexity bound or a general solver for the
block-diagonal auxiliary problem.
\section{A Generic Linear Decomposition Attack}
\label{sec:generic-lda}

We now show that the public-generator version of the tensor-action framework is
generically vulnerable to the linear decomposition method of Myasnikov and
Roman'kov~\cite{MyasnikovRomankov2014LDA}. The attack treats cubic tensors as vectors in
the $d^3$-dimensional space
\[
V:=\F_p^d\otimes\F_p^d\otimes\F_p^d
\]
and each tensor action as a linear endomorphism of $V$. For
$A,B,C\in\GL_d(\F_p)$, write
\[
\Psi(A,B,C):=A\otimes B\otimes C\in\GL(V),
\]
so that, under the natural tensor-product identification,
\[
\Psi(A,B,C)T=(A,B,C)\cdot T.
\]
Let $G_{\mathcal X}$ and $G_{\mathcal Y}$ be public finite generating sets for
$\mathcal X$ and $\mathcal Y$. From them one obtains public generating sets of
endomorphisms
\begin{align*}
\mathcal U_{\mathcal X}:={}&
\{\Psi(g,I,I),\Psi(I,g,I),\Psi(I,I,g):
g\in G_{\mathcal X}\cup G_{\mathcal X}^{-1}\},\\
\mathcal U_{\mathcal Y}:={}&
\{\Psi(h,I,I),\Psi(I,h,I),\Psi(I,I,h):
h\in G_{\mathcal Y}\cup G_{\mathcal Y}^{-1}\}.
\end{align*}
The monoids generated by these sets contain all actions from
$\mathcal X^3$ and $\mathcal Y^3$, respectively. Moreover, every element generated by
$\mathcal U_{\mathcal X}$ commutes with every element generated by
$\mathcal U_{\mathcal Y}$ because $\mathcal X$ and $\mathcal Y$ commute element-wise.

\begin{theorem}[Generic linear decomposition attack]
\label{thm:generic-linear-decomposition}
Assume that $\mathcal X,\mathcal Y\leq\GL_d(\F_p)$ are given by public finite
generating sets and commute element-wise. Given the public transcript
\[
T,\qquad
T_A=\Psi(A,B,C)T,\qquad
T_B=\Psi(D,E,F)T,
\]
where $A,B,C\in\mathcal X$ and $D,E,F\in\mathcal Y$, an adversary can recover
the shared tensor
\[
S=\Psi(D,E,F)T_A=\Psi(A,B,C)T_B
\]
using deterministic polynomial-time linear algebra over $\F_p$.
\end{theorem}

\begin{proof}
Let $\mathscr A$ be the endomorphism monoid generated by
$\mathcal U_{\mathcal X}$ and consider the orbit span
\[
W_{\mathcal X}(T)
:=\operatorname{span}_{\F_p}\{MT:M\in\mathscr A\}
\subseteq V.
\]
A basis of this space can be computed by iterative closure. Start with $T$, apply every
generator in $\mathcal U_{\mathcal X}$ to every current basis vector, and retain a new
vector whenever it is linearly independent of the vectors already retained. Repeating this
process terminates after at most $\dim_{\F_p}V=d^3$ basis vectors. Keeping the
corresponding words in the public generators produces known endomorphisms
$M_1,\dots,M_r\in\mathscr A$ such that
\[
\{M_1T,\dots,M_rT\}
\]
is a basis of $W_{\mathcal X}(T)$, with $r\leq d^3$.

Set $\alpha:=\Psi(A,B,C)$ and $\beta:=\Psi(D,E,F)$. Since
$\alpha\in\mathscr A$, the public tensor $T_A=\alpha T$ lies in
$W_{\mathcal X}(T)$. Gaussian elimination therefore gives public coefficients
$\lambda_1,\dots,\lambda_r\in\F_p$ satisfying
\[
T_A=\sum_{i=1}^r\lambda_i M_iT.
\]
The adversary applies the same known endomorphisms to Bob's public tensor and outputs
\[
\widehat S:=\sum_{i=1}^r\lambda_iM_iT_B.
\]
Every $M_i$ is generated by Alice-side actions and hence commutes with the Bob-side
endomorphism $\beta$. Consequently,
\begin{align*}
\widehat S
&=\sum_{i=1}^r\lambda_iM_i\beta T
=\beta\sum_{i=1}^r\lambda_iM_iT\\
&=\beta T_A
=\beta\alpha T
=S.
\end{align*}
The basis construction and coefficient recovery use linear algebra in dimension $d^3$
and are polynomial in the public input size. The Kronecker matrices need not be formed
explicitly, since each generator can be applied by a single mode-wise matrix
multiplication.
\end{proof}

For the encryption scheme of Section~\ref{sec:tensor-ko-lee}, Bob's ciphertext contains
$T_B=U$ and
\[
c=m\oplus\Hash(S).
\]
After computing $S$ by Theorem~\ref{thm:generic-linear-decomposition}, the adversary
recovers the plaintext as
\[
m=c\oplus\Hash(S).
\]
The attack does not recover $(A,B,C)$ or $(D,E,F)$ and is independent of the choice of
the tensor sampler. Thus, under the public finite-generator model used by the concrete
constructions in this paper, changing only the tensor distribution or the commuting
subgroups cannot yield a secure instantiation of the present linear-action template. The
construction-specific analyses in Section~\ref{sec:pitfalls-commuting} remain useful because
they expose simpler failure mechanisms and support the empirical illustrations, but they
are not needed for the generic insecurity conclusion.
\section{Discussion: Influence of Scaled-Block Structure}
\label{sec:discussion}
The generic attack in Section~\ref{sec:generic-lda} breaks the public finite-generator
framework, while Section~\ref{sec:pitfalls-commuting} identifies additional failure
mechanisms in the three concrete constructions. Here we isolate the leakage caused by a
common scaled subspace. This discussion does not mitigate the generic attack or supply a
secure instantiation.

\label{subsec:influence_scaled-block}
The block-diagonal reduction of Section~\ref{sec:pitfalls-commuting} shows that a
\emph{large} fixed subspace exposes lower-dimensional algebraic structure. It is therefore
natural to ask what remains
true when one only assumes the presence of a smaller common scaled subspace. The next
lemma records the corresponding structural consequence. It shows that such a subspace still
reduces the tensor-action problem to lower-dimensional matrix relations on the affected
slices, so the resulting problem is easier than the original one. However, for small $r$ this
reduction need not imply an efficient full break. Thus fixed or scaled subspaces should be
viewed as dangerous, while the extent of the danger depends on their size. The proof is
given in Appendix~\ref{proof:scaled-block-structure}.

Let $r\in\{1,\dots,d-1\}$ denote the dimension of the scaled subspace. Assume that
$A,B,C$ scale a common $r$-dimensional coordinate subspace; that is, after a suitable
(possibly public) choice of basis they admit the block forms
\begin{align}
A=\begin{pmatrix}\bar A & 0\\[2pt] 0 & aI_r\end{pmatrix},\qquad
B=\begin{pmatrix}\bar B & 0\\[2pt] 0 & bI_r\end{pmatrix},\qquad
C=\begin{pmatrix}\bar C & 0\\[2pt] 0 & cI_r\end{pmatrix},
\end{align}
where $\bar A,\bar B,\bar C\in \GL_{d-r}(\F_p)$, $a,b,c\in \F_p^\times$, and $I_r$ is
the $r\times r$ identity matrix. For compatibility with this decomposition, we write the
$k$-th slice $T_k$ of the tensor $T$ in block form as
\[
T_k=\begin{pmatrix}T_{k,1} & T_{k,2}\\ T_{k,3} & T_{k,4}\end{pmatrix}.
\]

\begin{lemma}
\label{lem:scaled-block-structure}
Given $T$ and $T'=(A,B,C)\cdot T$, where $A,B,C$ are sampled from subgroup
$\mathcal{X}$ whose elements all scale the same subspace, formally
\[
X_0 = \begin{pmatrix}\bar X_0 & 0\\[2pt] 0 & xI_r\end{pmatrix}, \qquad X_0\in \mathcal{X},
\]
then solving the action $(A,B,C)$ is reduced to solving $(\bar A,\bar B)$, as defined
above, from the systems indexed by $k\in\{d-r+1,\dots,d\}$:
\begin{align*}
&T'_{k,1} = c\,\bar A T_{k,1}\bar B^\top,\qquad \\
& T'_{k,2} = (cb)\,\bar A T_{k,2},\qquad  \\
& T'_{k,3} = (ca)\,T_{k,3}\bar B^\top,\qquad \\
& T'_{k,4} = (cab)\,T_{k,4}.
\end{align*}
\end{lemma}

\begin{remark}[Why a single-slice ``decomposition'' need not recover the true $(\bar A,\bar B)$]
\label{rem:scaled-block-not-unique}
At first glance, the relation
\begin{equation}
\label{eq:single-slice-decomp}
T'_{k,1}=c\,\bar A\,T_{k,1}\,\bar B^\top
\end{equation}
for a fixed $k\in\{d-r+1,\dots,d\}$ resembles a standard matrix equivalence
(or left--right action) instance: given matrices $U,V$, find $P,Q$ such that
$V=P\,U\,Q$. Over a field, such equivalence problems are algorithmically tractable:
two matrices are equivalent under $\mathrm{GL}\times\mathrm{GL}$ if and only if they have
the same rank, and canonical forms can be computed by elementary row and column
operations \cite{HornJohnson2013}. Consequently, if one ignores the membership
constraints $(A,B)\in \mathcal{X}$ and the coupling among different slices, it is easy to
produce some pair $(\widehat A,\widehat B)$ satisfying \eqref{eq:single-slice-decomp}.

However, such a pair need not coincide with the ground-truth $(\bar A,\bar B)$ used in
the tensor action, nor does it necessarily suffice to reconstruct the full action. The reason
is that $C$ mixes the family of matrices
\[
\{A T_k B^\top\}_{k=1}^d
\]
across different slices, so the relevant object is not a single decomposition
\eqref{eq:single-slice-decomp}, but rather the entire span generated by these matrices.
A decomposition obtained from one slice may instead produce a different family
\[
\{\widehat A T_k \widehat B^\top\}_{k=1}^d,
\]
and there is no reason for the two spans
\[
\mathrm{span}\{A T_k B^\top\}_{k=1}^d
\qquad\text{and}\qquad
\mathrm{span}\{\widehat A T_k \widehat B^\top\}_{k=1}^d
\]
to coincide. In that case, even if $(\widehat A,\widehat B)$ matches the observed block
for one or several slices, one still cannot surely recover
$\widehat C$ reproducing the full tensor action. Thus Lemma~\ref{lem:scaled-block-structure}
should be interpreted as evidence of \emph{structural leakage} caused by the scaled
subspace, rather than as a complete attack for every parameter choice.

In particular, when $r$ is small, the lemma shows only that the original problem reduces
to an easier lower-dimensional matrix problem; it does not show that the resulting problem
is itself easy to solve. The remark above highlights exactly this gap: local decompositions
on individual slices need not assemble into the true global tensor action. Accordingly,
small fixed or scaled subspaces still look dangerous, but whether they yield a practical full
break depends on the size of the subspace and on the residual difficulty of the reduced
problem.
\end{remark}
\section{Conclusion}

We introduced a Ko--Lee-style framework for public-key encryption based on tensor actions
and proved its formal correctness. The construction replaces conjugation by the natural
action of $\GL_d(\F_p)^3$ on cubic tensors and makes commuting-subgroup generation a
central instantiation step. We showed, however, that the resulting public finite-generator
framework admits a generic linear decomposition attack in the $d^3$-dimensional tensor
space. The attack recovers the shared tensor, and hence the plaintext, using polynomial-time
linear algebra without recovering either secret action.

We also cryptanalyzed several natural
commuting-subgroup constructions, including field-extension-based abelian subgroups,
block-diagonal commuting non-abelian subgroups, and tensor-product commuting subgroups,
and identified diagonal, fixed-subspace, tensor-separable, and lower-dimensional matrix
structure exposed by these choices. These construction-specific results range from
conditional action recovery or structural reduction to direct shared-tensor recovery; the
generic attack independently breaks all three public finite-generator instantiations.

The scaled-block analysis further describes the framework's structural leakage, but it does
not provide evidence of security. Under the public
finite-generator model considered here, changing the tensor sampler or commuting subgroups
cannot prevent the generic attack. The paper should therefore be read as a framework
proposal and its cryptanalysis, not as a deployable public-key encryption scheme. Any future
variant would need to depart from at least one of the public linear-action assumptions used
by the decomposition attack and would require a new security analysis.

\newpage
\bibliographystyle{splncs04}
\bibliography{main}

\vspace{0.5cm}

\appendix
\begin{center}
	\Large \textbf{Public Key Encryption on Tensor Isomorphism\\ \vspace{8pt} \textit{Appendices}}
\end{center}

\section{Ko--Lee Paradigm}
\label{sec:ko-lee-template}

We recall the basic Ko--Lee public-key encryption scheme introduced in
\cite{KoLeeCheonHanKangPark2000BraidPKE}. Let $G$ be a non-abelian group with efficiently solvable word problem (given $s,t\in G$, determine whether $s=t$), meaning easy recognition of different representations of the same group element, and let $X,Y \leq G$ be two \emph{commuting} subsets, i.e.,
\[
\forall x\in X,\ \forall y\in Y,\quad xy = yx.
\]
Let $\mathbf{p}\in G$ be public. Alice samples a secret $s\in X$ and publishes
\[
\mathbf{p}' := s\mathbf{p}s^{-1}.
\]
To send a message $m\in\{0,1\}^\ell$, Bob samples $r\in Y$ and computes
\[
\mathbf{p}'' := r\mathbf{p}r^{-1},
\qquad
K_B := r\mathbf{p}' r^{-1} = r(s\mathbf{p}s^{-1})r^{-1}.
\]
He sends the ciphertext
\begin{align*}
    c &:= m \oplus H(K_B),\\
    C &:= \bigl(\mathbf{p}'',\ c\bigr),
\end{align*}
where $H: G \rightarrow \{0,1\}^\ell$ is a cryptographic hash and $\oplus$ denotes bitwise exclusive-OR. Alice computes
\begin{align*}
    K_A &:= s\mathbf{p}'' s^{-1} = s(r\mathbf{p}r^{-1})s^{-1},\\
    m &:= c \oplus H(K_A).
\end{align*}
Correctness holds because $sr=rs$ implies $K_A = K_B$.

\section{Proofs of Construction Pitfalls}
\subsection{Proof of Lemma~\ref{lem:scaled-block-structure}}
\label{proof:scaled-block-structure}
We begin by viewing the tensor action slice by slice. Since
\[
C=\begin{pmatrix}\bar C & 0\\ 0 & cI_r\end{pmatrix},
\]
for any $k\in\{d-r+1,\dots,d\}$, the $k$-th row of $C$ scales the $k$-th coordinate, equivalently the $k$-th slice of the tensor, by $c$. Hence the slice action in Equation~\eqref{eq:slice-action} simplifies to
\begin{equation}
\label{eq:last_slices}
T'_k=c\cdot A T_k B^\top,\qquad k=d-r+1,\dots,d,
\end{equation}
so $c\cdot A T_k B^\top$ is directly known from $T'_k$, the $k$-th matrix in $T'$.

Next, write each slice $T_k$ in the same $(d-r)+r$ block shape to make it compatible with the shape of $A$ and $B$, where:
\[
T_k=\begin{pmatrix}T_{k,1} & T_{k,2}\\ T_{k,3} & T_{k,4}\end{pmatrix},
\qquad
T_{k,1}\in\mathbb{F}_p^{(d-r)\times(d-r)},\;
T_{k,2}\in\mathbb{F}_p^{(d-r)\times r},
\]
\[
T_{k,3}\in\mathbb{F}_p^{r\times(d-r)},\;
T_{k,4}\in\mathbb{F}_p^{r\times r}.
\]
Using the block forms of $A$ and $B$, we expand
\begin{align}
\label{eq:block_expand}
A T_k B^\top
& =
\begin{pmatrix}\bar A & 0\\ 0 & aI_r\end{pmatrix}
\begin{pmatrix}T_{k,1} & T_{k,2}\\ T_{k,3} & T_{k,4}\end{pmatrix}
\begin{pmatrix}\bar B^\top & 0\\ 0 & bI_r\end{pmatrix} 
\\
\\
& =
\begin{pmatrix}
\bar A T_{k,1}\bar B^\top & b\,\bar A T_{k,2}\\[4pt]
a\,T_{k,3}\bar B^\top & ab\,T_{k,4}
\end{pmatrix}.
\end{align}
Multiplying by $c$, we learn from each $k\in\{d-r+1,\dots,d\}$ the following four blocks:
\begin{align*}
&T'_{k,1} = c\,\bar A T_{k,1}\bar B^\top,\qquad
 T'_{k,2} = (cb)\,\bar A T_{k,2},\\
& T'_{k,3} = (ca)\,T_{k,3}\bar B^\top,\qquad
T'_{k,4} = (cab)\,T_{k,4},
\end{align*}
all of which can be explicitly read from $T'_k$.

\subsection{Proof of Lemma~\ref{lem:simul-commute}}
\label{proof:simul-commute}
Fix $h_1,h_2\in \mathcal{H}$ and set $k_1:=P^{-1}h_1P$ and $k_2:=P^{-1}h_2P$.
By assumption $k_1,k_2\in \mathcal{K}$ and $\mathcal{K}$ is commuting, hence $k_1k_2=k_2k_1$.
Conjugating back by $P$ gives
\[
h_1h_2
=
P k_1 P^{-1} P k_2 P^{-1}
=
P(k_1k_2)P^{-1}
=
P(k_2k_1)P^{-1}
=
h_2h_1,
\]
as required.

\subsection{Proof of Proposition~\ref{prop:diagonalizable-easy}}
\label{proof:diagonalizable-easy}
Let $T\in \mathbb{F}_p^{d\times d\times d}$ be public and let
\[
T_A := (A,B,C)\cdot T
\quad\text{with}\quad
A,B,C\in \mathcal{H}.
\]
Define the basis-changed tensors over $\mathbb{F}_{p^m}$:
\[
\widetilde{T} := (P^{-1},P^{-1},P^{-1})\cdot T,
\qquad
\widetilde{T}_A := (P^{-1},P^{-1},P^{-1})\cdot T_A.
\]
Write
\begin{align*}
&\widetilde{A}:=P^{-1}AP=\mathrm{diag}(a_1,\dots,a_d),\quad\\
&\widetilde{B}:=P^{-1}BP=\mathrm{diag}(b_1,\dots,b_d),\quad\\
&\widetilde{C}:=P^{-1}CP=\mathrm{diag}(c_1,\dots,c_d),
\end{align*}
with $a_i,b_j,c_k\in \mathbb{F}_{p^m}^\times$. Then for all $i,j,k\in[d]$,
\begin{equation}
\widetilde{T}_{A,ijk} = a_i\,b_j\,c_k\ \widetilde{T}_{ijk}.
\label{eq:diag-scaling}
\end{equation}
By the support hypothesis, there exists a pivot triple $(i_0,j_0,k_0)$ with
$\widetilde{T}_{i_0j_0k_0}\neq 0$, and the entries
$\widetilde{T}_{ij_0k_0}$, $\widetilde{T}_{i_0jk_0}$, and
$\widetilde{T}_{i_0j_0k}$ are nonzero for every $i,j,k$. Hence the diagonal entries
$(a_i)$, $(b_j)$, and $(c_k)$ can be recovered from
$(\widetilde{T},\widetilde{T}_A)$ in polynomial time up to scalars:
\[
(a_i,b_j,c_k)\mapsto (\lambda a_i,\mu b_j,(\lambda\mu)^{-1}c_k),
\qquad \lambda,\mu\in \mathbb{F}_{p^m}^\times.
\]
Consequently one can reconstruct a representative of the scalar-gauge class of
$(A,B,C)$ by conjugating the recovered diagonal matrices back with $P$, and this
representative induces the same action on $T$.

Since $A,B,C\in \mathcal{H}$, by simultaneous diagonalizability we may conjugate them by the same $P$ to diagonal matrices $\widetilde{A},\widetilde{B},\widetilde{C}$. Using the action definition and the fact that conjugation by $P$ corresponds to a basis change, we obtain Equation~\eqref{eq:diag-scaling}. For any triple $(i,j,k)$ with $\widetilde{T}_{ijk}\neq 0$, define the ratio
\[
R_{ijk} := \frac{\widetilde{T}_{A,ijk}}{\widetilde{T}_{ijk}} = a_i b_j c_k.
\]
Fix the pivot $(i_0,j_0,k_0)$ from the support hypothesis. Then
\[
\frac{R_{ij_0k_0}}{R_{i_0j_0k_0}} = \frac{a_i}{a_{i_0}}, \quad \frac{R_{i_0jk_0}}{R_{i_0j_0k_0}} = \frac{b_j}{b_{j_0}}, \quad \frac{R_{i_0j_0k}}{R_{i_0j_0k_0}} = \frac{c_k}{c_{k_0}}
\]
for every $i,j,k$. For example, set
\[
\widehat a_i:=\frac{R_{ij_0k_0}}{R_{i_0j_0k_0}},\qquad
\widehat b_j:=\frac{R_{i_0jk_0}}{R_{i_0j_0k_0}},\qquad
\widehat c_k:=R_{i_0j_0k}.
\]
Then $\widehat a_i\widehat b_j\widehat c_k=a_i b_j c_k$ for all $i,j,k$, so
$(\widehat A,\widehat B,\widehat C)$ is gauge-equivalent to
$(\widetilde A,\widetilde B,\widetilde C)$. Finally, a gauge-equivalent action is
reconstructed by conjugating these diagonal matrices back with $P$.

\subsection{Proof of Proposition~\ref{prop:block-reduction}}
\label{proof:block-reduction}
By Lemma~\ref{lem:tensor-action-composition}, composition of tensor actions is componentwise:
\[
(A,B,C)\circ (D,E,F)=(AD,\;BE,\;CF).
\]
Since
\[
AD=
\begin{pmatrix}
\bar A&0\\
0&\bar D
\end{pmatrix},
\qquad
BE=
\begin{pmatrix}
\bar B&0\\
0&\bar E
\end{pmatrix},
\qquad
CF=
\begin{pmatrix}
\bar C&0\\
0&\bar F
\end{pmatrix},
\]
for every $i\in\{1,\dots,m\}$, the first $m$ output slices satisfy
\[
\bigl[(A,B,C)\cdot((D,E,F)\cdot T)\bigr]_i
=
\sum_{k=1}^m \bar C_{ik}
\begin{pmatrix}
\bar A T_{k,1}\bar B^\top & \bar A T_{k,2}\bar E^\top\\
\bar D T_{k,3}\bar B^\top & \bar D T_{k,4}\bar E^\top
\end{pmatrix}.
\]
where we divide $T$ into blocks:
\[
T=
\begin{pmatrix}
T_{k,1}& T_{k,2}\\
T_{k,3}& T_{k,4}
\end{pmatrix}
\]

Thus the $i$-th slice breaks into four $m\times m$ blocks:
\begin{align}
\label{eq:block1-new}
\text{(block 1)}\quad & \sum_{k=1}^m \bar C_{ik}\,\bar A T_{k,1}\bar B^\top,\\
\label{eq:block2-new}
\text{(block 2)}\quad & \sum_{k=1}^m \bar C_{ik}\,\bar A T_{k,2}\bar E^\top,\\
\label{eq:block3-new}
\text{(block 3)}\quad & \sum_{k=1}^m \bar C_{ik}\,\bar D T_{k,3}\bar B^\top,\\
\label{eq:block4-new}
\text{(block 4)}\quad & \sum_{k=1}^m \bar C_{ik}\,\bar D T_{k,4}\bar E^\top.
\end{align}
We show that each block is either already known or reducible to the auxiliary mixed-product problem.

\medskip
\noindent\textbf{Block 1.}
For $i\le m$, the first $m$ rows of $C$ are exactly the rows of $\bar C$, while the lower-right identity block in $B$ fixes the first block-column. Hence the first block of the $i$-th slice of $(A,B,C)\cdot T$ is
\[
\sum_{k=1}^m \bar C_{ik}\,\bar A T_{k,1}\bar B^\top,
\]
which is known and exactly equals Equation~\eqref{eq:block1-new}. Therefore block 1 is already known from the public tensor $(A,B,C)\cdot T$.

\medskip
\noindent\textbf{Block 2.}
For $i\le m$, the second block of the $i$-th slice of $(A,B,C)\cdot T$ is
\[
U_i:=\sum_{k=1}^m \bar C_{ik}\,\bar A T_{k,2},
\]
because the lower-right identity block in $B$ acts trivially on the second block-column (which can be verified through block-wise matrix multiplication). On the other hand, since the first $m$ rows of $F$ form the identity, the $i$-th slice of
$(D,E,F)\cdot T$ uses only $T_i$, and its second block equals
\[
V_i:=T_{i,2}\bar E^\top.
\]
Thus, denoting the second block of each matrix in the tensor, which is also known, 
\[
S_i:=T_{i,2}\qquad (1\le i\le m),
\]
the public data  give exactly an instance of the auxiliary mixed-product problem:
\[
U_i=\sum_{k=1}^m \bar C_{ik}\,\bar A S_k,
\qquad
V_i=S_i\bar E^\top.
\]
Solving that instance returns
\[
W_i=\sum_{k=1}^m \bar C_{ik}\,\bar A S_k \bar E^\top
=
\sum_{k=1}^m \bar C_{ik}\,\bar A T_{k,2}\bar E^\top,
\]
which is precisely block 2 in \eqref{eq:block2-new}.

\medskip
\noindent\textbf{Block 3.}
The analysis of Block 3 is symmetric to that of Block 2. Let
\[
R_i:=T_{i,3}^\top \qquad (1\le i\le m).
\]
From $(A,B,C)\cdot T$, the third block of the $i$-th slice is
\[
\sum_{k=1}^m \bar C_{ik}\,T_{k,3}\bar B^\top.
\]
Taking transpose gives
\[
\left(\sum_{k=1}^m \bar C_{ik}\,T_{k,3}\bar B^\top\right)^\top
=
\sum_{k=1}^m \bar C_{ik}\,\bar B T_{k,3}^\top
=
\sum_{k=1}^m \bar C_{ik}\,\bar B R_k.
\]
Likewise, from $(D,E,F)\cdot T$, the third block of the $i$-th slice is
\[
\bar D T_{i,3},
\]
so after transpose we obtain
\[
(\bar D T_{i,3})^\top = T_{i,3}^\top \bar D^\top = R_i \bar D^\top.
\]
Therefore block 3 is again an instance of the same auxiliary mixed-product problem, now with
\[
\bar A \rightsquigarrow \bar B,\qquad \bar E^\top \rightsquigarrow \bar D^\top,\qquad
S_i \rightsquigarrow R_i.
\]
which is exactly the same problem setup as Block 2. Solving this instance gives
\[
\sum_{k=1}^m \bar C_{ik}\,\bar B R_k \bar D^\top,
\]
and transposing back yields
\[
\left(\sum_{k=1}^m \bar C_{ik}\,\bar B R_k \bar D^\top\right)^\top
=
\sum_{k=1}^m \bar C_{ik}\,\bar D T_{k,3}\bar B^\top,
\]
which is exactly block 3 in \eqref{eq:block3-new}. 

\medskip
\noindent\textbf{Block 4.}
The fourth block is
\[
\sum_{k=1}^m \bar C_{ik}\,\bar D T_{k,4}\bar E^\top.
\]
This has the same algebraic form as block 2 after a public permutation of the block decomposition. Indeed, let
\[
P=\begin{pmatrix}
0&I_m\\
I_m&0
\end{pmatrix}.
\]
Conjugating the first two modes by $P$ swaps the top and bottom block rows and columns in every slice, so the lower-right block $T_{k,4}$ is moved to the upper-right position. From another perspective, if we change the direction how we slice a tensor into a list of matrices, then the original block, which is diagonal to the first block, will thereby be adjacent. Under this public relabeling, the roles of the identity block and active block in the definitions of $X$ and $Y$ are exchanged, and the transformed version of block 4 becomes identical in shape to block 2. Hence the same auxiliary mixed-product solver applies verbatim. Equivalently, one may regard this as repeating the block-2 argument after swapping the two $m$-dimensional coordinate halves.

Combining the four cases, blocks 2, 3, and 4 are reducible to the auxiliary mixed-product problem, while block 1 is already known. Therefore the full tensor
\[
(A,B,C)\cdot((D,E,F)\cdot T)
\]
can be computed from the public data once the auxiliary mixed-product problem is solvable.

\subsection{Proof of Proposition~\ref{prop:coarse-fine-flattening-attack}}
\label{proof:coarse-fine-flattening-attack}
Let
\[
U=V\otimes W\cong \F_p^{n^2}
\]
with $V=W=\F_p^n$. We index the standard basis of $U$ by pairs
\[
(i,\alpha)\in[n]\times[n],
\]
where $i$ is the coarse index and $\alpha$ is the fine index. Thus a tensor
$T\in U^{\otimes 3}$ has entries
\[
T_{(i,\alpha),(j,\beta),(k,\gamma)}.
\]

Define the coarse--fine flattening
\[
\mathsf{Flat}:U^{\otimes 3}\longrightarrow \F_p^{n^3\times n^3}, \quad \mathsf{Flat}(T)_{(i,j,k),(\alpha,\beta,\gamma)}
:=
T_{(i,\alpha),(j,\beta),(k,\gamma)}.
\]
This is simply a reindexing of coordinates, so it is an invertible linear map. Write
\[
P:=\mathsf{Flat}(T),
\qquad
P_A:=\mathsf{Flat}(T_A),
\qquad
P_B:=\mathsf{Flat}(T_B).
\]

We first compute the effect of Alice's action. Since
\[
A=M_1\otimes I_n,
\qquad
B=M_2\otimes I_n,
\qquad
C=M_3\otimes I_n,
\]
their entries satisfy
\[
A_{(u,\alpha),(i,\delta)}=(M_1)_{u i}\,\delta_{\alpha\delta},
\]
and similarly for $B$ and $C$. Therefore
\begin{align*}
(T_A)_{(u,\alpha),(v,\beta),(w,\gamma)}
&=
\sum_{i,j,k=1}^n
(M_1)_{u i}(M_2)_{v j}(M_3)_{w k}
T_{(i,\alpha),(j,\beta),(k,\gamma)}.
\end{align*}
Hence, under the flattening $\mathsf{Flat}$,
\[
P_A=LP,
\qquad
L:=M_1\otimes M_2\otimes M_3\in \GL_{n^3}(\F_p).
\]

Next compute the effect of Bob's action. Since
\[
D=I_n\otimes N_1,
\qquad
E=I_n\otimes N_2,
\qquad
F=I_n\otimes N_3,
\]
their entries satisfy
\[
D_{(i,\alpha),(x,\delta)}=\delta_{ix}(N_1)_{\alpha\delta},
\]
and similarly for $E$ and $F$. Therefore
\begin{align*}
(T_B)_{(i,\alpha),(j,\beta),(k,\gamma)}
&=
\sum_{\delta,\varepsilon,\zeta=1}^n
(N_1)_{\alpha\delta}
(N_2)_{\beta\varepsilon}
(N_3)_{\gamma\zeta}
T_{(i,\delta),(j,\varepsilon),(k,\zeta)}.
\end{align*}
Thus
\[
P_B=PR^\top,
\qquad
R:=N_1\otimes N_2\otimes N_3\in \GL_{n^3}(\F_p).
\]

Applying the same flattening to the shared tensor
\[
T_{AB}=(A,B,C)\cdot T_B
\]
gives
\[
\mathsf{Flat}(T_{AB})=LPR^\top.
\]

If $P$ is invertible, then the shared tensor is recovered immediately from the public data:
\[
\mathsf{Flat}(T_{AB})=P_A P^{-1} P_B.
\]

Invertibility of $P$ is not necessary. Let $r=\mathrm{rank}(P)$. By Gaussian elimination,
one can compute invertible matrices $U_0,V_0\in \GL_{n^3}(\F_p)$ such that
\[
P=U_0
\begin{pmatrix}
I_r & 0\\
0 & 0
\end{pmatrix}
V_0.
\]
Set
\[
G:=V_0^{-1}
\begin{pmatrix}
I_r & 0\\
0 & 0
\end{pmatrix}
U_0^{-1}.
\]
Then $PGP=P$, so $G$ is an inner generalized inverse of $P$. Consequently,
\[
P_A G P_B
=
(LP)G(PR^\top)
=
L(PGP)R^\top
=
LPR^\top
=
\mathsf{Flat}(T_{AB}).
\]
Thus the adversary can compute the flattened shared tensor directly from the public
transcript and then apply $\mathsf{Flat}^{-1}$ to recover $T_{AB}$ itself. This breaks the
corresponding tensor-action Ko--Lee instantiation in polynomial time.

\section*{AI Tool Disclosure}
The authors used ChatGPT for assistance with literature review, language editing, and
implementation of empirical experiments. All mathematical claims, proofs, experimental
results, and final manuscript content were independently checked and edited by the authors.

\end{document}